\documentclass[11pt]{article}%
\usepackage{amsfonts}
\usepackage{amsmath}
\usepackage{amssymb}
\usepackage{graphicx}%
\providecommand{\U}[1]{\protect\rule{.1in}{.1in}}
\newtheorem{theorem}{Theorem}
\newtheorem{acknowledgement}[theorem]{Acknowledgement}

\newtheorem{corollary}[theorem]{Corollary}

\newtheorem{definition}[theorem]{Definition}
\newtheorem{example}[theorem]{Example}

\newtheorem{lemma}[theorem]{Lemma}
\newtheorem{notation}[theorem]{Notation}

\newtheorem{proposition}[theorem]{Proposition}

\newenvironment{proof}[1][Proof]{\noindent\textbf{#1.} }{\ \rule{0.5em}{0.5em}}
\begin{document}

\title{A Phase Space Representation of the Metaplectic Group}
\author{Maurice de Gosson\thanks{maurice.de.gosson@univie.ac.at}\\University of Vienna\\Faculty of Mathematics (NuHAG)\\1090 Vienna, AUSTRIA\\and\\Austrian Academy of Sciences\\Acoustics Research Institute\\1010, Vienna, AUSTRIA}
\maketitle

\begin{abstract}
The symplectic group $\operatorname*{Sp}(n)$ acts on phase space
$\mathbb{R}_{x}^{n}\times\mathbb{R}_{p}^{n}$ while the unitary representation
of its double cover, $\operatorname*{Mp}(n)$, the metaplectic group, acts on
functions defined on $\mathbb{R}_{x}^{n}$. We will construct an extension
$\widetilde{\operatorname{Mp}}(n)$ of $\operatorname*{Mp}(n)$ acting on square
integrable functions on $\mathbb{R}_{x}^{n}\times\mathbb{R}_{p}^{n}$. This is
performed using previous results of ours involving explicit expressions of the
twisted Weyl symbols of metaplectic operators and Bopp pseudodifferential
operators, which are phase space extensions of the usual Weyl operators..

\end{abstract}

\section{Introduction}

The metaplectic group $\operatorname*{Mp}(n)$ has a long history, staring with
Andr\'{e} Weil's study in number theory and has benefitted from the
contribution of many authors (Buslaev \cite{Buslaev}, Shale \cite{Shale},
Leray \cite{Leray}, Reiter \cite{Reiter}, Maslov \cite{Maslov}, to cite a
few). While the major use of $\operatorname*{Mp}(n)$ in quantization problems,
there has recently been a resurgence of interest of this topic in
time-frequency analysis and.

Technically speaking, the metaplectic group is a connected group of unitary
operators acting on the Hilbert space of square integrable functions
satisfying the exactness of the sequence%
\[
0\longrightarrow\{\pm I\}\longrightarrow\operatorname*{Mp}(n)\longrightarrow
\operatorname*{Sp}(n)\longrightarrow0
\]
where $\operatorname*{Sp}(n)$ is the symplectic group: $\operatorname*{Mp}(n)$
is thus a double cover of $\operatorname*{Sp}(n)$. While the latter acts on
phase space $\mathbb{R}_{z}^{2n}\equiv\mathbb{R}_{x}^{n}\times\mathbb{R}%
_{p}^{n}$(or, in TFA, on the time-frequency domain),while $\operatorname*{Mp}%
(n)$ acts on functions defined on. 

In the present work we define and study an extension of metaplectic operators
acting unitarily on , and forming a group $\widetilde{\operatorname{Mp}}(n)$.
Thus, $\widetilde{\operatorname{Mp}}(n)$ that is on functions defined on the
symplectic phase space itself thus restoring a certain symmetry in the action
domains of both groups. This will be achieved by using results from \cite{LMP}
where we studied the Weyl symbols of metaplectic operators, and our
construction of a phase space pseudo-differntial calculus ("Bopp calculus")
based on the Bopp shifts \cite{Bopp} initiated in \cite{Birkbis,CPDE}.

\begin{notation}
The phase space $\mathbb{R}_{z}^{2n}\equiv\mathbb{R}_{x}^{n}\times
\mathbb{R}_{p}^{n}$ is equipped with the standard symplectic form
$\sigma(z,z^{\prime})=Jz\cdot z^{\prime}$, $J=%
\begin{pmatrix}
0 & I\\
-I & 0
\end{pmatrix}
$. The standard symplectic group $\operatorname*{Sp}(n)$ is the group of all
automorphisms $S$ of $\mathbb{R}_{z}^{2n}$ such that $SJS^{T}=S^{T}JS=J$.
\end{notation}

\section{Metaplectic Operators and their Weyl symbols}

\subsection{Definition by quadratic Fourier transforms}

For related studies and details see\cite{Folland,Birk,Leray}.

We are following here Let $W=(P,L;Q)$ be a quadratic form on $\mathbb{R}%
_{x}^{n}\times\mathbb{R}_{x}^{n}$ of the type
\begin{equation}
W(x,x^{\prime})=\tfrac{1}{2}Px\cdot x-Lx\cdot x^{\prime}+\tfrac{1}%
{2}Qx^{\prime}\cdot x^{\prime}\nonumber
\end{equation}
\textit{with} \ $P=P^{T}$ , \ $Q=Q^{T}$ , \textit{and }\ $\det L\neq0$. We
will call such a quadratic form a generating function because of the following
property. each such $Q$ determines a unique $S_{W}\in\operatorname*{Sp}(n)$
such that
\[
(x,p)=S_{W}(x^{\prime},p^{\prime})\Longleftrightarrow p=\partial
_{x}W(x,x^{\prime})\text{ and }p^{\prime}=-\partial_{x^{\prime}}W(x,x^{\prime
});
\]
a straightforward calculation shows that
\begin{equation}
S_{W}=%
\begin{pmatrix}
L^{-1}Q\smallskip & L^{-1}\\
PL^{-1}Q-L^{T} & L^{-1}P
\end{pmatrix}
_{\text{.}} \label{swlq}%
\end{equation}
Observe that $\det L^{-1}\neq0$ (such a symplectic matrix is said to be
"free"). Conversely, every free symplectic matrix%
\[
S=%
\begin{pmatrix}
A\smallskip & B\\
C & D
\end{pmatrix}
\text{ \ ., \ }\det B\neq0
\]
corresponds a unique generating function, namely%

\begin{equation}
W(x,x^{\prime})=\tfrac{1}{2}DB^{-1}x\cdot x-B^{-1}x\cdot x^{\prime}+\tfrac
{1}{2}B^{-1}x^{\prime}\cdot x^{\prime}.\nonumber
\end{equation}

\begin{definition}
The metaplectic group $\operatorname*{Mp}(n)$ is the group of unitary
operators in $L^{2}(\mathbb{R}_{x}^{n}))$ generated by the quadratic Fourier
integral operators%
\begin{equation}
\widehat{S}_{W,m}f(x)=\left(  \tfrac{1}{2\pi i\hbar}\right)  ^{n/2}%
\Delta(W)\int_{\mathbb{R}^{n}}e^{\frac{i}{\hbar}W(x,x^{\prime})}f(x^{\prime
})dx^{\prime}\text{;} \label{swm1}%
\end{equation}
where
\begin{equation}
\Delta(W)=i^{m}\sqrt{|\det L|} \label{swm2}%
\end{equation}
the integer $m$ )"Maslov index") corresponding to a choice of $\arg\det L$:%
\begin{equation}
m\pi\equiv\arg\det L\text{ \ \ }\operatorname{mod}2\pi\text{.} \label{swm3}%
\end{equation}

\end{definition}

\begin{example}
\label{ex1}Let $S_{\alpha}=%
\begin{pmatrix}
\cos\alpha\smallskip & \sin\alpha\\
-\cos\alpha & \sin\alpha
\end{pmatrix}
$, $\alpha\notin\pi\mathbb{Z}$. The generating function is
\[
W_{\alpha}(x,xx^{\prime})=\frac{1}{2\sin\alpha}(\cos\alpha(x^{2}+x^{\prime
2})-2xx^{\prime})
\]
and hence%
\[
\widehat{S}_{W_{\alpha},m}f(x)=\left(  \tfrac{1}{2\pi i\hbar}\right)
^{1/2}i^{[\alpha/\pi]}\sqrt{\frac{1}{2\sin\alpha}}\int_{\mathbb{R}^{n}%
}e^{\frac{i}{\hbar}W_{\alpha}(x,x^{\prime})}f(x^{\prime})dx^{\prime}\text{.}%
\]

\end{example}

The following properties of metaplectic operators are essential:

\begin{theorem}
\label{Thm1}(i) Every $\widehat{S}\in\operatorname{Mp}(n)$ is the product of
exactly two quadratic Fourier integral operators: $\widehat{S}=\widehat{S}%
_{W,m}\widehat{S}_{W,^{\prime}m^{\prime}}$ and (ii) the natural projection
$\pi^{\operatorname{Mp}}:\operatorname{Mp}(n)\longrightarrow\operatorname{Sp}%
(n)$ is defined \ by
\begin{equation}
\pi^{\operatorname{Mp}}(\widehat{S}_{W,m})=S_{W}\text{ \ },\text{ \ }%
\pi^{\operatorname{Mp}}(\widehat{S})=S_{W}S_{W^{\prime}}; \label{pisw1}%
\end{equation}
(iii) The inverse of $\widehat{S}_{W,m}$ is
\begin{equation}
\widehat{S}_{W,m}^{-1}=\widehat{S}_{W^{\prime},m^{\prime}}\text{ \ ,
\ }W^{\prime}(x,x^{\prime})=-W(x^{\prime},x)\text{ , }m^{\prime}=n-m;
\label{swinv}%
\end{equation}
(iv) If $\widehat{S}_{W,m})=\widehat{S}_{W^{\prime},m^{\prime}}\widehat{S}%
^{\prime\prime}$ with $W=(P^{\prime},L^{\prime},Q^{\prime})$, $W^{\prime
\prime}=(P^{\prime\prime},L^{\prime\prime},Q^{\prime\prime})$ then
\[
m=m^{\prime}+m^{\prime\prime}-\operatorname*{Inert}("+Q^{\prime})\text{
\ }\operatorname{mod}2\pi
\]
where $\operatorname*{Inert}(P"+Q^{\prime})$ is the index of inertia of the
symmetric matrix $P"+Q^{\prime}$.
\end{theorem}

\begin{proof}
See \cite{Birk}, Ch.7, \cite{Mpn}, \cite{Leray}, Ch.1.
\end{proof}

Some readers, in particular those coming from the time-frequency community,
might be more familiar with the following presentation of the metaplectic
group: we note that each $\widehat{S}_{W,m}$ can be factored as a product of
three types of elementary operator, namely%
\begin{equation}
\widehat{S}_{W,m}=\widehat{V}_{-P}\widehat{M}_{L,m}\widehat{J}\widehat{V}_{-Q}
\label{swplq}%
\end{equation}
where
\begin{equation}
\widehat{V}_{-P}f(x)=e^{\tfrac{i}{2\hbar}\langle Px,x\rangle}f(x)\text{ \ ,
\ }\widehat{M}_{L,m}f(x)=i^{m}\sqrt{|\det L|}f(Lx) \label{vpmlm}%
\end{equation}
and $\widehat{J}=\widehat{S}_{(0,I,0),0}$ which is basically the Fourier
transform:
\begin{equation}
\widehat{J}f(x)=\left(  \tfrac{1}{2\pi i\hbar}\right)  ^{n/2}\int%
_{\mathbb{R}^{n}}e^{-\frac{i}{\wp}\langle x,x^{\prime}\rangle}f(x^{\prime
})dx^{\prime}=i^{-n/2}Ff(x). \label{jichap}%
\end{equation}
It follows that the set of all operators $\widehat{V}_{-P}$, $\widehat{M}%
_{L,m}$ together with $\widehat{J}$ generate $\operatorname{Mp}(n).$ The
projections of these operators are
\[
\pi^{\operatorname{Mp}}(\widehat{V}_{-P})=V_{-P}=%
\begin{pmatrix}
I\smallskip & 0\\
P & I
\end{pmatrix}
\text{ },\text{ }\pi_{-P}^{\operatorname{Mp}}\text{\ }\widehat{M}_{L,m}%
)=M_{L}=%
\begin{pmatrix}
L^{-1}\smallskip & 0\\
0 & L^{T}%
\end{pmatrix}
\]
and $\pi_{-P}^{\operatorname{Mp}}(\widehat{J})=J=%
\begin{pmatrix}
0\smallskip & I\\
-I & 0
\end{pmatrix}
.$

\subsection{The twisted Weyl symbol of $\protect\widehat{S}_{W,m}$}

Let $\widehat{A}$ be a continuous operator $\mathcal{S}(\mathbb{R}_{x}%
^{n})\longrightarrow\mathcal{S}^{\prime}(\mathbb{R}_{x}^{n})$; in view of
Schwartz's kernel theorem there exists $K\in\mathcal{S}^{\prime}%
(\mathbb{R}_{x}^{n}\times\mathbb{R}_{x}^{n})$ such that%
\[
\widehat{A}f(x)=\int_{\mathbb{R}^{n}}K(x,y)f(y)dy
\]
(the integral being viewed in the sense of distributions). By definition, the
Weyl symbol of $\widehat{A}$ is the distribution $a$ defined by the Fourier
inversion formula%
\begin{equation}
a(x,p)=\int_{\mathbb{R}^{n}}e^{-\frac{i}{\hbar}p\cdot y}K(x+\tfrac{1}%
{2}y,x-\tfrac{1}{2}y)dy. \label{finv2}%
\end{equation}
Inverting this formula yields
\begin{equation}
K(x,y)=\left(  \tfrac{1}{2\pi\hbar}\right)  ^{n}\int_{\mathbb{R}^{n}}%
e^{\frac{i}{\hbar}p\cdot(x-y)}a(\tfrac{1}{2}(x+y),p)dp \label{KA9}%
\end{equation}
hence the familiar expression
\begin{equation}
\widehat{A}f(x)=\left(  \tfrac{1}{2\pi\hbar}\right)  ^{n}\int_{\mathbb{R}%
^{2n}}e^{\frac{i}{\hbar}p\cdot(x-y)}a(\tfrac{1}{2}(x+y),p)f(y)dydp
\label{apsigood}%
\end{equation}
valid under adequate conditions on $a$ and $f$.

There are several ways to express Weyl operators in integral form; in our
context it will be useful to use harmonic decomposition \cite{Birk,Horm}
\begin{equation}
\widehat{A}=\left(  \tfrac{1}{2\pi\hbar}\right)  ^{n}\int_{\mathbb{R}^{2n}%
}a_{\sigma}(z)\widehat{T}(z)z. \label{bochner}%
\end{equation}
Here $a_{\sigma}$ -- the \textit{twisted Weyl symbol} of $\widehat{A}$ -- is
the symplectic Fourier transform
\begin{equation}
a_{\sigma}(z)=F_{\sigma}a(z)=\left(  \tfrac{1}{2\pi\hbar}\right)  ^{n}%
\int_{\mathbb{R}^{n2}}e^{-\tfrac{i}{\hbar}\sigma(z,z^{\prime})}a(z^{\prime
})dz^{\prime} \label{SFT}%
\end{equation}
and $\widehat{T}(z_{0})$ the Heisenberg--Weyl outplacement operator:%
\begin{equation}
\widehat{T}(z_{a})f(x)=e^{\frac{i}{\hbar}p_{0}x-\tfrac{1}{2}p_{0}x_{0}%
,)}f(x-x_{0}) \label{HW}%
\end{equation}
which is a variant of the shift operator used in time-frequency analysis.

An operator $\widehat{S}\in\operatorname{Mp}(n)$ de facto satisfies the
conditions of Schwartz's kernel theorem and thus \textit{de facto} has a Weyl
symbol. To describe the latter we introduce the symplectic Cayley transform
\cite{LMP,Birk} of $S\in\operatorname*{Sp}(n)$ satisfying the condition
$\det(S-I)\neq0$ is invertible; it defined by%
\begin{equation}
M_{S}=\frac{1}{2}J(S+I)(S-I)^{-1}=\frac{1}{2}J+J(S-I)^{-1}. \label{cayley}%
\end{equation}
The symplectic Cayley transform is symmetric (\cite{Birk}, \S 4.3.2):
$M_{S}=M_{S}^{T}$ and we have the inversion formula%
\begin{equation}
M_{S^{-1}}=-M_{S} \label{invms}%
\end{equation}
and we have the inversion formula%
\begin{equation}
S=(M_{S}-\tfrac{1}{2}J)^{-1}(M_{S}+\tfrac{1}{2}J). \label{minv}%
\end{equation}

\begin{theorem}
\label{Thm2}Let $\widehat{S}=\widehat{S}_{W,m}$ be a quadratic Fourier
integral operator such that $\det(S_{W}-I)\neq0$. (i) The twisted the Weyl
symbol $a^{W,m}$ of is the function
\begin{equation}
a_{\sigma}^{W,m}(z)=\frac{i^{\nu}}{\sqrt{|\det(S_{W}-I)|}}e^{\frac{i}{2\hbar
}M_{W}z\cdot z} \label{asigmaw1}%
\end{equation}
where $M_{W}=M_{S_{W}}$ is the symplectic Cayley transform of $S_{W}$. (iii)
The integer $\nu$ is the Conley--Zehnder index%
\begin{equation}
\nu\equiv m-\operatorname*{Inert}W_{xx}\text{ \ }\operatorname{mod}4
\label{cz1}%
\end{equation}
where $\operatorname*{Inert}W_{xx}$ is the index of inertia of the Hessian of
the quadratic form $x\longmapsto W(x,x)W$; the Conley--Zehnder index
corresponds to a choice of $\arg\det(S-)$.%
\begin{equation}
\arg\det(S-)\equiv(-\nu+n)\pi.\text{ \ }\operatorname{mod}2\pi. \label{argdet}%
\end{equation}

\end{theorem}

\begin{proof}
See \cite{LMP,Birk}.. For a review of the Conley--Zehnder index and its
relation of the Maslov index see \cite{JMPA}.
\end{proof}

Note that if $W=(PL,Q)$ then%
\begin{equation}
W_{xx}=P+Q-L-L^{T} \label{wplq}%
\end{equation}
and
\begin{equation}
\det(S_{W}-I)=(-1)^{n}(\det L^{-1})\det P+Q-L-L^{T}). \label{detSW}%
\end{equation}

The operator $\widehat{S}_{W,m}$ is thus given by
\begin{equation}
\widehat{S}_{W,m}=\left(  \frac{1}{2\pi\hbar}\right)  ^{n}\frac{i^{\nu}}%
{\sqrt{|\det(S_{W}-I)|}}\int_{\mathbb{R}^{2n}}e^{\frac{i}{2\hbar}M_{W}z\cdot
z}\widehat{T}(z)dz; \label{s1}%
\end{equation}
this can be written in several different ways; for instance%
\begin{equation}
\widehat{S}_{W,m}=\left(  \frac{1}{2\pi\hbar}\right)  i^{\nu}\sqrt{|\det
(S_{W}-I)|}\int_{\mathbb{R}^{2n}}e^{\frac{i}{2\hbar}\sigma(Sz,z)}%
\widehat{T}((S-I)z)dz \label{s2}%
\end{equation}
that is, equivalently,
\begin{equation}
\widehat{S}_{W,m}=\left(  \frac{1}{2\pi\hbar}\right)  i^{\nu}\sqrt{|\det
(S_{W}-I)|}\int_{\mathbb{R}^{2n}}\widehat{T}(Sz)\widehat{T}(-z)dz. \label{s3}%
\end{equation}

\begin{example}
\label{ex2}As in Example \ref{ex1} let $S_{\alpha}=%
\begin{pmatrix}
\cos\alpha\smallskip & \sin\alpha\\
-\cos\alpha & \sin\alpha
\end{pmatrix}
$, $\sin\alpha\notin\pi\mathbb{Z}$. We have, for $\alpha\notin2\pi\mathbb{Z}%
$,
\begin{gather*}
M_{\alpha}=\frac{1}{2}J%
\acute{}%
+J(S_{\alpha}-I)^{-1}=^{-}=\frac{1}{2}%
\begin{pmatrix}
\cot\frac{\alpha}{2}\smallskip & 0\\
0 & \cot\frac{\alpha}{2}%
\end{pmatrix}
\\
\det(S_{\alpha}-I)=%
\begin{vmatrix}
\cos\alpha\smallskip-1 & \sin\alpha\\
-\cos\alpha & \sin\alpha-1
\end{vmatrix}
=4\sin^{2}\frac{\alpha}{2}%
\end{gather*}
hence the twisted symbol$a_{\sigma}^{\alpha}$ of $\widehat{S}_{\alpha}is$
given by
\[
a_{\sigma}^{\alpha}(z)=\frac{i^{^{[\alpha/\pi]-n}}}{\sqrt{|\det(S_{W}-I)|}%
}\exp\left(  \frac{i}{4\hbar}(x^{2}+p^{2})\cot\frac{\alpha}{2}\right)
\]

\end{example}

\subsection{ The general case}

The following Lemma complements part (i) of Theorem \ref{Thm1}.

\begin{lemma}
Every $\widehat{S}\in\operatorname{Mp}(n)$ can be written as a a product
$\widehat{S}=\widehat{S}_{W,m}\widehat{S}_{W,^{\prime}m^{\prime}}$ with
$\det(S_{W}-I)\neq0$ and and $\det(S_{W^{\prime}}-I)\neq0$.
\end{lemma}

\begin{proof}
In view of the factorization result (\ref{swplq}) we have
\begin{equation}
\widehat{S}=\widehat{V}_{-P}\widehat{M}_{L,m}\widehat{J}\widehat{V}%
_{-(Q+P^{\prime}=}\widehat{M}_{L^{\prime},m^{\prime}}\widehat{J}%
\widehat{V}_{-Q^{\prime}} \label{fact1}%
\end{equation}
and the conditions $\det(S_{W}-I)\neq0\det(S_{W^{\prime}}-I)\neq0$ are
equivalent formula (\ref{detSW}))%
\begin{equation}
\det P+Q-L-L^{T}).\det P^{\prime}+Q^{\prime}-L^{\prime}-L^{\prime T})\neq0.
\label{fact2}%
\end{equation}
Factorization (\ref{fact1}) of $\widehat{S}$ does not change if we replace
simultaneously $Q$ with $Q+\lambda I$ and $P^{\prime}$ with $Q$ with
$P^{\prime}-\lambda I$ for some $\lambda\in\mathbb{R}$. Then choose $\lambda$
such that (\ref{fact2}) holds.
\end{proof}

Theorem \ref{Thm2} implies that:

\begin{corollary}
\label{cor1}Let $\widehat{S}\in\operatorname{Mp}(n)$ be such that $\det
(S_{W}-I)\neq0$. \ If $\widehat{S}=\widehat{S}_{W,m}\widehat{S}_{W,^{\prime
}m^{\prime}}$ with $\det(S_{W}-I)\neq0$ and and $\det(S_{W^{\prime}}-I)\neq0$.
Then the twisted Weyl symbol $a^{\widehat{S}}$ of $\widehat{S}$ is given by%
\[
a^{\widehat{S}}(z)=\frac{i^{\nu+\nu^{\prime}+\operatorname*{sign}(M))}}%
{\sqrt{|\det(S-I)}}e^{\frac{i}{2\hbar}M_{S}z\cdot z}%
\]
where $\operatorname*{sign}(M)$ is the signature of $M=M_{W}+M_{W^{\prime}}$
and $\nu$, $\nu^{\prime}$ the Conley--Zehnder indices of $\widehat{S}_{W,m}$
and $\widehat{S}_{W,^{\prime}m^{\prime}}$, respectively.
\end{corollary}

\begin{proof}
See \cite{LMP} and \cite{Birk}, Ch.7.
\end{proof}

\section{Extension of $\operatorname{Mp}(n)$ to Phase Space}

\subsection{Bopp pseudodifferential operators}

Bopp operators are extensions of Weyl operators to phase space; We will be
following the approach given in \cite{Birkbis}, Ch.19; also see \cite{CPDE}
for the study of spectral properties of these operators.

Bopp \cite{Bopp} essentially consists in replacing the usual Schr\"{o}dinger
"quantization rules" $x\longrightarrow x,$, $p\longrightarrow-i\hbar
\partial_{x}$\ with the more symmetric operators\
\begin{equation}
x\longrightarrow\widetilde{x}=x+\tfrac{1}{2}i\hbar\partial_{p}\text{ \ ,
\ }p\longrightarrow\ \widetilde{p}=p-\tfrac{1}{2}i\hbar\partial_{x}
\label{bopp}%
\end{equation}
which act, not on functions defined on $\mathbb{R}^{n}$, but on functions
defined on the phase space $\mathbb{R}^{2n}$. To make this rigorous, we begin
by defining a phase-space version of the Heisenberg--Weyl displacement
operators (\ref{HW}) by setting, for $F\in\mathcal{S}^{\prime}(\mathbb{R}%
^{2n})$,%
\[
\widetilde{T}(z_{0})F(z)=e^{-\frac{i}{\hbar}\sigma(z,z_{0})}F(z-\tfrac{1}%
{2}z_{0}).
\]
These phase space operators obey relations similar to those%
\begin{align}
\widehat{T}(z_{0})\widehat{T}(z_{1})  &  =e^{-\frac{i}{\hbar}\sigma
(z_{0},z_{1})}\widehat{T}(z_{1})\widehat{T}(z_{0})\label{noco1}\\
\widehat{T}(z_{0}+z_{1})  &  =e^{-\tfrac{i}{2\hbar}\sigma(z_{0},z_{1}%
)}\widehat{T}(z_{0})\widehat{T}(z_{1}) \label{noco2}%
\end{align}
satisfied by the operators $\widehat{T}(z_{0})$, namely:%
\begin{align}
\widetilde{T}(z_{0}+z_{1})  &  =e^{-\frac{i}{2\hbar}\sigma(z_{0},z_{1}%
)}\widetilde{T}(z_{0})\widetilde{T}(z_{1})\label{a15}\\
\widetilde{T}(z_{1})\widetilde{T}(z_{0})  &  =e^{-\frac{i}{\hbar}\sigma
(z_{0},z_{1})}\widetilde{T}(z_{0})\widetilde{T}(z_{1}). \label{b15}%
\end{align}
An essential observation is that $\widehat{T}(z_{0})\psi($ and $\widetilde{T}%
(z_{0})$ are intertwined by the cross-Wigner transform%
\[
W(f,g)=\left(  \tfrac{1}{2\pi\hbar}\right)  ^{n}\int_{\mathbb{R}^{n}}%
e^{-\frac{i}{\hbar}p\cdot y}f(x+\tfrac{1}{2}y)\overline{g(x-\tfrac{1}{2}%
y)}dy.
\]
For $f,g\in L^{2}(\mathbb{R}^{n})$ we have
\begin{equation}
W(\widehat{T}(z_{0})f,g)=\widetilde{T}(z_{0})W(f,g). \label{intertikde}%
\end{equation}

Recalling formula (\ref{bochner}) which gives the Weyl operator in terms of
the symplectic Fourier transform (\ref{SFT}=of its symbol, we define:

\begin{definition}
Let $a\in\mathcal{S}^{\prime}(\mathbb{R}^{2n})$; the Bopp operator
$\widetilde{A}=\operatorname*{Op}\nolimits_{\mathrm{Bopp}}(a)$ with symbol $a$
is defined by%
\begin{equation}
\widetilde{A}F(z)=\left(  \tfrac{1}{2\pi\hbar}\right)  ^{n}\int_{\mathbb{R}%
^{2n}}a_{\sigma}(z_{0})\widetilde{T}(z_{0})F(z)dz_{0} \label{atilde}%
\end{equation}
for $F\in\mathcal{S}(\mathbb{R}^{2n})$.
\end{definition}

Bopp operators have two following crucial properties, which we summarize in
the theorem below:

\begin{theorem}
(i) The Bopp operator $\widetilde{A}=\operatorname*{Op}%
\nolimits_{\mathrm{Bopp}}(a)$ is a continuous operator $\mathcal{S}%
(\mathbb{R}^{2n})\longrightarrow\mathcal{S}^{\prime}(\mathbb{R}^{2n})$ hence a
Weyl operator; its symbol is given by%
\begin{equation}
\widetilde{a}(z,\zeta)=a(x-\tfrac{1}{2}\zeta_{p},p+\tfrac{1}{2}\zeta_{x})
\label{atilde15}%
\end{equation}
where $\zeta=(\zeta_{x},\zeta_{p})$; (ii) Let $(f,g\in L^{2}(\mathbb{R}^{n})$;
we have he transform
\begin{equation}
\widetilde{A}W(f,g)=W\widehat{A}f,g). \label{uffi}%
\end{equation}
(iii) The mapping $U_{g}f=(2\pi\hbar)^{n/2}W(f,g)$is a partial isometry
$L^{2}(\mathbb{R}^{n})\longrightarrow L^{2}(\mathbb{R}^{2n}).$intertwining the
Weyl operator $\widehat{A}$ and the corresponding Bopp operator $\widetilde{A}%
$:%
\begin{equation}
\widetilde{A}U_{g}=U_{g}\widehat{A}. \label{afu}%
\end{equation}

\end{theorem}

\begin{proof}
See \cite{Birkbis}, Chapters 18--19 for detailed proofs. Notice that
(\ref{afu}) follows from the relation (\ref{intertikde}). Property (iii)
follows using Moyal's identity.,.
\end{proof}

\subsection{Application to metaplectic operators}

Let $\widehat{S}\in\in\operatorname{Mp}(n)$ be a$\det(S-I)\neq0$Viewed as a
Weyl operator it is explicitly given by the Bochner integral
\begin{equation}
\widehat{S}f(x)=\left(  \frac{1}{2\pi\hbar}\right)  ^{n}\frac{i^{\nu}}%
{\sqrt{|\det(S-I)|}}\int_{\mathbb{R}^{2n}}e^{\frac{i}{2\hbar}M_{SW}z_{0}\cdot
z_{0}}\widehat{T}(z_{0})dz_{0}%
\end{equation}
and its Bopp representation is thus obtained by replacing $\widehat{T}(z)$
with $\widetilde{T}(z)$:
\begin{equation}
\widetilde{S}_{W,m}F(z)=\left(  \frac{1}{2\pi\hbar}\right)  ^{n}\frac{i^{\nu}%
}{\sqrt{|\det(S_{W}-I)|}}\int_{\mathbb{R}^{2n}}e^{\frac{i}{2\hbar}M_{_{W}%
S}z_{0}\cdot z_{0}}\widetilde{T}(z_{0})F(z)dz_{0} \label{sf}%
\end{equation}
more generally (Corollary \ref{cor1}),%
\begin{equation}
\widetilde{S}F(z)=\frac{i^{\nu+\nu^{\prime}+\operatorname*{sign}(M_{S}))}%
}{\sqrt{|\det(S-I)}}\int_{\mathbb{R}^{2n}}e^{\frac{i}{2\hbar}M_{S}z_{0}\cdot
z_{0}}\widetilde{T}(z_{0})F(z)dz_{0}. \label{sfbis}%
\end{equation}
In view of the inversion formula (\ref{swinv}=) the operators $\widetilde{S}%
_{W,m}$ are invertible, and
\begin{equation}
\widetilde{S}_{W,m}^{-1}=\widetilde{S}_{W^{\prime},m^{\prime}}\text{ \ ,
\ }W^{\prime}(x,x^{\prime})=-W(x^{\prime},x)\text{ , }m^{\prime}=n-m.
\label{swinvbis}%
\end{equation}

\begin{definition}
The group of operators generated by the $\widetilde{S}_{W,m}$, $\det
(S_{W}-I)\neq0$, is denoted by $\widetilde{\operatorname{Mp}}(n)$ and called
the extended metaplectic group-.
\end{definition}

We are going to prove that $\widetilde{\operatorname{Mp}}(n)$ is a group of
unitary operator on $L^{2}(\mathbb{R}^{2n})$. To prove this we need the
following lemma:

\begin{lemma}
\label{lem}Let $(f_{j})_{j}$ and $(g_{j})_{j}$ be orthonormal bases of
$L^{2}(\mathbb{R}^{n})$. Then system of vectors $((2\pi\hbar)^{n}W(f_{j}%
,g_{k}))_{jk}$ is an orthonormal basis of $L^{2}(\mathbb{R}^{2n})$.
\end{lemma}

\begin{proof}
See \cite{Birk}, Ch.9, Thm. 4.4.2. The orthonormality of the vectors
$(2\pi\hbar)^{n/2}W(f_{j},g_{k})$ follows from Moyal's identity%
\begin{equation}
(W(f,g)|W(f^{\prime},g^{\prime}))_{L^{2}(\mathbb{R}^{2n})}=\left(  \frac
{1}{2\pi\hbar}\right)  ^{n}i(f|f^{\prime})_{L^{2}(\mathbb{R}^{n}).}%
\overline{(g|g^{\prime})_{L^{2}(\mathbb{R}^{n}).}}. \label{Moyal}%
\end{equation}

\end{proof}

\begin{theorem}
Let $\widehat{S}\in\operatorname*{Mp}(n)$ have projection $S\in
\operatorname*{Sp}(n)$such that $\det(S-I)\neq0$. (i) Let $f,g\in
L^{2}(\mathbb{R}^{n})$. We have%
\begin{equation}
\widetilde{S}W(f,g)=W(\widehat{S}f,g)\label{swf}%
\end{equation}
and hence $|\widetilde{S}$
\begin{equation}
||\widetilde{S}W(f,g)||_{L^{2}(\mathbb{R}^{2n})}=||f||_{L^{2}(\mathbb{R}%
^{n}).}||g||_{L^{2}(\mathbb{R}^{n}).}.\label{sw}%
\end{equation}
(ii) Let $F\in L^{2}(\mathbb{R}^{2n})$. We have \
\begin{equation}
||\widetilde{S}F||_{L^{2}(\mathbb{R}^{2n})}=||F||_{L^{2}(\mathbb{R}^{2n})}%
^{2}.\label{SF}%
\end{equation}
hence the correspondence $F\longmapsto\widetilde{S}F\in$ is unitary on
$L^{2}(\mathbb{R}^{2n})$.
\end{theorem}

\begin{proof}
(i) Formula (\ref{swf}) follows from (\ref{afu}) taking into account
definition (\ref{uffi}). Formula (\ref{sw}) follows from Moyal's identity
(\ref{Moyal}). (ii) Let $(f_{j})_{j}$ and $(g_{j})_{j}$ be orthonormal bases
of $L^{2}(\mathbb{R}^{n})$.In view of Lemma \ref{lem} we can write
\[
F=(2\pi\hbar)^{n/2}\sum\lambda_{jk}W(f_{j},g_{k})
\]
and hence, by (\ref{swf}),
\[
\widetilde{S}F=(2\pi\hbar)^{n/2}\sum\lambda_{jk}\widetilde{S}W(f_{j}%
,g_{k}=(2\pi\hbar)^{n/2}\sum\lambda_{jk}W(\widehat{S}f_{j},g_{k}.
\]
Since $(f_{j})_{j}$ is an orthonormal basis so is $(\widehat{S}f_{j})_{j}$
(because $\widehat{S}$ is unitary) and hence%
\[
||\widetilde{S}F||_{L^{2}(\mathbb{R}^{2n})}^{2}=(2\pi\hbar)^{n}\sum
\lambda_{jk}^{2}=||F||_{L^{2}(\mathbb{R}^{2n})}^{2}%
\]
which proves formula (\ref{SF}).
\end{proof}

The operators $\widehat{S}$ can be expressed in alternative ways involving the
displacements $\widetilde{T}(z)$:

\begin{proposition}
Let $\widehat{S}\in\operatorname*{Mp}(n)$ have projection $S\in
\operatorname*{Sp}(n)$such that $\det(S-I)\neq0$. We have
\begin{equation}
\widetilde{S}=\left(  \frac{1}{2\pi\hbar}\right)  ^{n}i^{\nu}\sqrt
{|\det(S-I)|}\int_{\mathbb{R}^{2n}}e^{-\frac{i}{2\hbar}\sigma(Sz,z)}%
\widetilde{T}((S-I)z)dz \label{alfa2}%
\end{equation}
that is, as
\begin{equation}
\widehat{S}_{\nu}=\left(  \frac{1}{2\pi\hbar}\right)  ^{n}i^{\nu}\sqrt
{|\det(S-I)|}\int_{\mathbb{R}^{2n}}\widetilde{T}(Sz)\widetilde{T}%
(-z)dz\text{.} \label{alfa1}%
\end{equation}

\end{proposition}

\begin{proof}
We have
\[
\frac{1}{2}J(S+I)(S-I)^{-1}=\frac{1}{2}J+J(S-I)^{-1}%
\]
hence, in view of the antisymmetry of $J$,%
\[
M_{S}z\cdot z=J(S-I)^{-1}z\cdot z=\sigma((S-I)^{-1}z,z)
\]
Performing the change of variables $z\longmapsto(S-I)^{-1}z$ we can write%
\begin{align*}
\int_{\mathbb{R}^{2n}}e^{\frac{i}{2\hbar}M_{S}zz}\widehat{T}(z)dz  &
=\int_{\mathbb{R}^{2n}}e^{\frac{i}{2\hbar}\sigma(z,(S-I)z)}\widetilde{T}%
((S-I)z)dz\\
&  =\int_{\mathbb{R}^{2n}}e^{-\frac{i}{2\hbar}\sigma(Sz,z)}\widetilde{T}%
((S-I)z)dz
\end{align*}
hence (\ref{alfa1}). Taking into account the relation (\ref{noco1}) we have%
\[
\widetilde{T}((S-I)z)=e^{-\tfrac{i}{2\hbar}\sigma(Sz,z)}\widetilde{T}%
(Sz)\widetilde{T}(-z)
\]
and formula (\ref{alfa2}) follows.
\end{proof}

\section{Applications}

\subsection{The Feichtinger Algebra.}

The Wigner formalism allows to define an algebra of functions on configuration
space well adapted for the study of phase space quantum mechanics. This
algebra -- the Feichtinger algebra --- is usually defined in terms of the
short-time Fourier transform (STFT) \cite{Gro}, but we will rather use the
Wigner transform. (see our presentation in \cite{Birkbis}).

The Feichtinger algebra, of which we give here a simple (non-traditional)
definition is a particular case of thew more general notion of Feichtinger's
modulation spaces \cite{Hans1,Hans2,Hans3}; for a comprehensive study see
Gr\"{o}chenig's treatise \cite{Gro}. These spaces play an important role in
time-frequency analysis, but are yet underestimated in quantum mechanics.

\begin{definition}
The Feichtinger algebra $S_{0}(\mathbb{R}^{n}$ consists of all function
$\psi\in L^{2}(\mathbb{R}^{n})$ such that $W\psi\in L^{1}(\mathbb{R}^{2n})$.
\end{definition}

It is not immediately clear from this definition that $S_{0}(\mathbb{R}^{n}$
\ is a vector space. However \cite{Gro}, (\cite{Birkbis}, Ch.16):

\begin{proposition}
(i) We have $\psi\in S_{0}(\mathbb{R}^{n})$ if and only if there exists one
window $\phi$ such that $W(\psi,\phi)\in L^{1}(\mathbb{R}^{2n})$, in which
case we have $W(\psi,\phi)\in L^{1}(\mathbb{R}^{2n})$ for all windows
$\phi\phi\in\mathcal{S}(\mathbb{R}^{n})$; ; (ii) If $W(\psi,\phi)\in
L^{1}(\mathbb{R}^{2n})$ then both $\psi$ and $\phi$ are in $S_{0}%
(\mathbb{R}^{n})$; (iii) The functions $\psi\longmapsto||\psi||_{\phi,S_{0}}$
($\phi\in\mathcal{S}(\mathbb{R}^{n})$) defined by
\[
||\psi||_{\phi,S_{0}}=||W(\psi,\phi)||_{L^{1}(\mathbb{R}^{2n})}%
\]
are equivalent norms on $S_{0}(\mathbb{R}^{n}),$ which is a Banach space for
the apology thus defined. (iv) $S_{0}(\mathbb{R}^{n})$ is an algebra for both
usual (pointwise) multiplication and convolution.
\end{proposition}

We have the inclusions%
\begin{equation}
\mathcal{S}(\mathbb{R}^{n})\subset S_{0}(\mathbb{R}^{n})\subset C^{0}%
(\mathbb{R}^{n})\cap L^{1}(\mathbb{R}^{n})\cap L^{2}(\mathbb{R}^{n}).
\label{inclo}%
\end{equation}

\begin{proposition}
\label{169}Let $\psi\in S_{0}(\mathbb{R}^{n})$. We have (i) $\widehat{S}%
\psi\in S_{0}(\mathbb{R}^{n})$ for every $\widehat{S}\in\operatorname*{Mp}%
(n)$,; (ii) $\widehat{T}(z_{0})\psi\in S_{0}(\mathbb{R}^{n})$ for every
$z_{0}\in\mathbb{R}^{2n}$ . (iii) We have $\lim_{|x|\rightarrow\infty}\psi=0$
hence $\psi$ is bounded.
\end{proposition}

\begin{proof}
(Cf. \cite{Birkbis}, Ch. 16). (i) We have $\psi\in S_{0}(\mathbb{R}^{n})$ if
and only $\psi\in L^{2}(\mathbb{R}^{n})$ and $W\psi\in L^{1}(\mathbb{R}^{2n}%
)$. The property follows from the covariance relation $W(\widehat{S}%
\psi)=W\psi\circ S^{-1}$ where $S\in\operatorname*{Sp}(n)$ is the projection
of $\widehat{S}$. (ii) Follows similarly from the translation property
$W(\widehat{T}(z_{0})\psi)=W\psi(z-z_{0})$. (iii) \ Since $\psi$ is continuous
it boundedness follows from $\lim_{z\rightarrow\infty}\psi=0$. Since
$S_{0}(\mathbb{R}^{n})$ is invariant by Fourier transform in view of (i) , we
have $F^{-1}\psi\in S_{0}(\mathbb{R}^{n})$; now $S_{0}(\mathbb{R}^{n})\subset
L^{1}(\mathbb{R}^{n})$ hence $\psi=F(F^{-1}\psi)$ has limit $0$ at infinity in
view of Riemann--Lebesgue's lemma.
\end{proof}

The following result describes a characterization of the Feichtinger algebra
in terms of the phase space metaplectic operators:

\begin{proposition}
(i) We have $f\in S_{0}(\mathbb{R}^{n})$ if and only if $\widetilde{S}%
W(f,g)\in L^{1}(\mathbb{R}^{2n})$ \ for some $\widetilde{S}$;(ii=) when this
is the case we have $\widetilde{S}W(f,g)\in L^{1}(\mathbb{R}^{2n})$ for all
$\widetilde{S}$.
\end{proposition}

\begin{proof}
(i) Assume that $f\in S_{0}(\mathbb{R}^{n})$, then $\widetilde{S}%
W(f,g)=W(\overline{S}f,g)\in L^{1}(\mathbb{R}^{2n})$. If, conversely,
$\widetilde{S}W(f,g)=\in L^{1}(\mathbb{R}^{2n})$ then $\overline{S}f\in
S_{0}(\mathbb{R}^{n})$ and hence $f\in S_{0}(\mathbb{R}^{n})$ in view of the
metaplectic invariance of the Feichtinger algebra. (ii) is clear in view of
the argument above.
\end{proof}

\subsubsection{Asymptotics for $\hbar\rightarrow0$}

Recall the method of stationary phase \cite{Fedo,Leray}: let Let $a\in
C_{0}^{\infty}({\mathbb{R}}^{n},{\mathbb{C)}}$ and $\varphi\in a\in
C_{0}^{\infty}({\mathbb{R}}^{n},{\mathbb{R)}}$. We assume that $\varphi$ has
only non-degenerate critical points in $\operatorname{supp}(a)$,
\textit{i.e}.,
\[
\partial_{x}\varphi(x_{c})=0\quad{\text{and}}\quad\det\operatorname*{Hess}%
(x_{c})\neq0
\]
($\operatorname*{Hess}(x_{c})$ the Hessian matrix at $x_{c}$). Consider the
integral
\[
I(\lambda)=\int_{{\mathbb{R}}^{n}}e^{i\lambda\varphi(x)}a(x)\,dx,\qquad
\lambda>0.
\]
For $\lambda\rightarrow+\infty$ we have the asymptote approximation:
\begin{equation}
I(\lambda)=(2\pi/\lambda)^{n/2}\sum_{x_{c}\in{\mathrm{Crit}}(\varphi
)}e^{i\lambda\varphi(x_{c})}\,e^{i\frac{\pi}{4}\operatorname*{sign}%
\operatorname*{Hess}(x_{c}))}\,\frac{a(x_{c})}{|\det\operatorname*{Hess}%
\varphi(x_{c})|^{1/2}}+O(\lambda^{-n/2-1}).\label{phastat}%
\end{equation}
Let us set $\lambda=1/\hbar$ and apply this formula to the formula
\begin{equation}
\widetilde{S}F(z)=\left(  \frac{1}{2\pi\hbar}\right)  ^{n}\frac{i^{\nu}}%
{\sqrt{|\det(S-I)|}}\int_{\mathbb{R}^{2n}}e^{\frac{i}{2\hbar}M_{S}z_{0}\cdot
z_{0}}\widetilde{T}(z_{0})F(z)dz_{0}%
\end{equation}
with $F\in C_{0}^{\infty}({\mathbb{R}}^{n},{\mathbb{C)}}$.

\begin{theorem}
For $\det(S-I)\det(S-I)\neq0$ We have the following asymptotic expression of
$\widetilde{S}F(z)$for $\hbar\rightarrow0$:%
\begin{multline}
\widetilde{S}F(z)=\frac{^{2^{-n}}i^{\nu-\frac{1}{2}\operatorname*{sign}M_{S}}%
}{\det(S-I)\sqrt{\det(S+I)}}\label{asympt}\\
\times\exp e\left(  \frac{i}{2\hbar}M_{-S}z\cdot z_{J}\right)  F(z-M_{S}%
^{-1}Jz)++O(\hbar^{n}).
\end{multline}

\end{theorem}

\begin{proof}
Replacing $x$ with $z$ and $n$ with $2n$ we have%
\[
\varphi(z_{0})=\frac{1}{2}M_{S}z_{0}\cdot z_{0}-Jz\cdot z_{0}%
\]
hence $\partial_{z}\varphi(z_{c})=0$ if and only if
\[
z_{c}=M_{S}^{-1}Jz=2/S-I)(S+I)^{-1}%
\]
and hence%
\begin{align*}
\varphi(z_{c})  & =\frac{1}{2}M_{S}(M_{S}^{-1}J)z\cdot M_{S}^{-1}J-Jz\\
& =\frac{1}{2}Jz\cdot M_{S}^{-1}Jz-Jz\cdot M_{S}^{-1}z\\
& =\frac{1}{2}JM_{S}^{-1}Jzz\\
& =2M_{-S}%
\end{align*}
We have $\operatorname*{Hess}(\varphi(z_{c})=M_{S}$. hence
\begin{equation}
\det\operatorname*{Hess}((\varphi(z_{c}))=\det M_{S}=2^{-n}\det(S+I)\det
r(S-I)^{-1}\label{dethessian}%
\end{equation}
Collecting all these results, a straightforward calculation leads to formula
(\ref{asympt}).
\end{proof}

\begin{acknowledgement}
This work has been financed by the Austrian Research Foundation FWF (Quantum
Austria PAT 2056623).
\end{acknowledgement}

\end{document}